\documentclass[a4paper,11pt]{article}
\usepackage{enumitem}
\usepackage{float}
\usepackage[margin=2.5cm]{geometry}
\usepackage[caption=false]{subfig}
\usepackage{graphicx}
\usepackage[labelfont=bf,labelsep=period]{caption}
\usepackage{times}
\usepackage{url}
\usepackage{xspace}
\usepackage{numbertabbing}

\setlist[]{topsep=2pt,partopsep=2pt,parsep=2pt,itemsep=2pt}

\floatstyle{ruled}
\newfloat{algo}{htbp}{algo}
\floatname{algo}{Algorithm}

\usepackage{amsmath} %%% better math %%%
\allowdisplaybreaks[2]          % but allow some broken displays

\usepackage{amssymb} %%% symbols %%%

\usepackage{amsthm} %%% Theorem environments with `amsthm' %%%
\newtheoremstyle{plain-boldhead}% name
  {\topsep}%      Space above
  {\topsep}%      Space below
  {\itshape}%     Body font
  {}%         Indent amount (empty = no indent, \parindent = para indent)
  {\bfseries}% Thm head font
  {.}%        Punctuation after thm head
  { }%     Space after thm head: " " = normal space; \newline = linebreak
  {\thmname{#1}\thmnumber{ #2}\thmnote{ (\bfseries #3)}}%    Thm head spec
\newtheoremstyle{definition-boldhead}% name
  {\topsep}%      Space above
  {\topsep}%      Space below
  {\normalfont}% Body font
  {}%         Indent amount (empty = no indent, \parindent = para indent)
  {\bfseries}% Thm head font
  {.}%        Punctuation after thm head
  { }%     Space after thm head: " " = normal space; \newline = linebreak
  {\thmname{#1}\thmnumber{ #2}\thmnote{ (\bfseries #3)}}%    Thm head spec
\theoremstyle{plain-boldhead}
\newtheorem{theorem}{Theorem}

\newtheorem{lemma}[theorem]{Lemma}

\theoremstyle{definition-boldhead}
\newtheorem{definition}{Definition}

\floatstyle{ruled}
\newfloat{algo}{htbp}{algo}
\floatname{algo}{Algorithm}

\newcommand{\etal}{\textit{et~al.}\@}

\newcommand{\str}[1]{\textsc{#1}}
\newcommand{\var}[1]{\textit{#1}}
\newcommand{\op}[1]{\textsl{#1}}
\newcommand{\msg}[2]{\ensuremath{\ifempty{#2} [\str{#1}] \else [\str{#1}, {#2}] \fi}}

\newcommand{\false}{\textsc{false}\xspace}

\newcommand{\BZ}{\ensuremath{\mathbb{Z}}\xspace}
\newcommand{\BN}{\ensuremath{\mathbb{N}}\xspace}

\newcommand{\CA}{\ensuremath{\mathcal{A}}\xspace}

\newcommand{\CF}{\ensuremath{\mathcal{F}}\xspace}

\newcommand{\CK}{\ensuremath{\mathcal{K}}\xspace}
\newcommand{\CM}{\ensuremath{\mathcal{M}}\xspace}

\newcommand{\CP}{\ensuremath{\mathcal{P}}\xspace}
\newcommand{\CQ}{\ensuremath{\mathcal{Q}}\xspace}

\newcommand{\CX}{\ensuremath{\mathcal{X}}\xspace}

\newcommand{\ceil}[1]{\lceil #1 \rceil}

\newcommand\viewNumber{\var{viewNumber}\xspace}
\newcommand\curView{\var{curView}\xspace}

\newcommand\node{\var{node}\xspace}
\newcommand\justify{\var{justify}\xspace}
\newcommand\type{\var{type}\xspace}
\newcommand\highQC{\var{highQC}\xspace}
\newcommand\prepareQC{\var{prepareQC}\xspace}
\newcommand\precommitQC{\var{precommitQC}\xspace}
\newcommand\commitQC{\var{commitQC}\xspace}
\newcommand\lockedQC{\var{lockedQC}\xspace}

\newcommand\newview{\var{new-views}\xspace}
\newcommand\prepares{\var{prepare-votes}\xspace}
\newcommand\precommits{\var{precommit-votes}\xspace}
\newcommand\commits{\var{commit-votes}\xspace}
\newcommand\qc{\var{qc}\xspace}
\newcommand\sig{\var{sig}\xspace}

\newcommand\parent{\var{parent}\xspace}
\newcommand\cmd{\var{cmd}\xspace}
\newcommand\height{\var{height}\xspace}
\newcommand\votes{\var{votes}\xspace}

\newcommand\createLeaf{\op{createLeaf}\xspace}

\newcommand\isQuorum{\op{isQuorum}\xspace}

\newcommand\update{\op{update}\xspace}
\newcommand\updateQCHigh{\op{updateQCHigh}\xspace}
\newcommand\onCommit{\op{onCommit}\xspace}
\newcommand\onReceiveProposal{\op{onReceiveProposal}\xspace}
\newcommand\onReceiveVote{\op{onReceiveVote}\xspace}
\newcommand\onPropose{\op{onPropose}\xspace}
\newcommand\getLeader{\op{getLeader}\xspace}
\newcommand\onBeat{\op{onBeat}\xspace}
\newcommand\onNextSyncView{\op{onNextSyncView}\xspace}
\newcommand\onReceiveNewView{\op{onReceiveNewView}\xspace}
\newcommand\execute{\op{execute}\xspace}
\newcommand\QC{\op{QC}\xspace}

\newcommand{\onecom}{2L1C}
\newcommand{\on}{MSP-HotStuff}
\newcommand{\off}{Counting-HotStuff}
\newcommand{\mix}{MSP-Replicas}
\newcommand{\form}{MBF-HotStuff}

\newcommand{\rvector}{\boldsymbol{r}}

\newcommand{\lvector}{\boldsymbol{\lambda}}
\newcommand{\evector}{\boldsymbol{e}}
\newcommand{\yvector}{\boldsymbol{y}}

\newcommand{\xvector}{\boldsymbol{x}}

\newcommand{\zerovector}{\boldsymbol{0}}

\def \ifempty#1{\def\temp{#1} \ifx\temp\empty }

\graphicspath{ {./images/} }
\newcommand{\dotparagraph}[1]{\paragraph*{#1.}}

\begin{document}

\title{\bf Consensus Beyond Thresholds:\\ Generalized Byzantine Quorums Made Live}

\author{Orestis Alpos$^1$\\
  University of Bern\\
  \url{orestis.alpos@inf.unibe.ch}
\and Christian Cachin$^1$\\
  University of Bern\\
  \url{cachin@inf.unibe.ch}
}

\date{}

\maketitle

\footnotetext[1]{Institute of Computer Science, University of Bern,
  Neubr\"{u}ckstrasse 10, 3012 CH-Bern, Switzerland.}

\begin{abstract}\noindent
  Existing Byzantine fault-tolerant (BFT) consensus protocols address only
  \emph{threshold failures}, where the participating nodes fail
  independently of each other, each one fails equally likely, and the
  protocol's guarantees follow from a simple bound on the \emph{number} of
  faulty nodes.  With the widespread deployment of Byzantine consensus in
  blockchains and distributed ledgers today, however, more sophisticated
  trust assumptions are needed.

  This paper presents the first implementation of BFT consensus with
  generalized quorums.  It starts from a number of generalized trust
  structures motivated by practice and explores methods to specify and
  implement them efficiently.  In particular, it expresses the trust
  assumption by a monotone Boolean formula (MBF) with threshold operators
  and by a monotone span program (MSP), a linear-algebraic model for computation.
  An implementation of HotStuff BFT
  consensus using these quorum systems is described as well and compared to
  the existing threshold model. Benchmarks with HotStuff running on up to 
  40 replicas demonstrate that the MBF specification incurs no significant
  slowdown, whereas the MSP expression affects latency and throughput
  noticeably due to the involved computations.
\end{abstract}

\section{Introduction}
Trust assumptions are a fundamental part of secure distributed computing protocols. On one hand, they capture the limits of a protocol's safety properties, thus characterizing the domains in which it may be deployed safely.
But on the other hand, they also impose limits on the potential of the protocol and, in some sense, the  expressiveness and freedom of the parties, thus restricting the domains in which the protocol will be deployed.
\emph{Byzantine quorum systems}~\cite{DBLP:journals/dc/MalkhiR98} are the key abstraction for capturing the trust assumptions in distributed protocols where parties may behave maliciously. A Byzantine quorum system (BQS) is defined as a set of \emph{quorums}, where a quorum is a set of parties that is sufficient to execute a particular task. A BQS is closely related with a \emph{fail-prone system}, which contains the sets of parties that are tolerated to fail in an execution, through the following intersection property: any two quorums must intersect in a set of parties that is not expected to fail. Thus, BQS formalize the expected Byzantine failures and allow reasoning about the resilience of protocols using them.

We refer to a BQS that is allowed to contain arbitrary quorums as a \emph{generalized BQS}, in contrast to a \emph{threshold BQS} that defines quorums only by their cardinality.
Generalized BQS have been intensely explored in the literature~\cite{DBLP:journals/dc/MalkhiR98}. For example, Malkhi \etal~\cite{DBLP:journals/siamcomp/MalkhiRW00} study their load and availability, Hirt and Maurer~\cite{DBLP:journals/joc/HirtM00} use a very related notion for secure multiparty computation, Junqueira \etal~\cite{DBLP:journals/dc/JunqueiraMHP10} explore an equivalent formalization in terms of survivor sets, and Warns \etal~\cite{DBLP:conf/srds/WarnsFH06} introduce a generalized model that unifies multiple such failure models.

Nevertheless, these works approach generalized BQS mainly from a theoretical perspective. 
When considering practical, state-of-the-art distributed protocols with Byzantine faults, especially state-machine replication (SMR) protocols in the blockchain space, 
one notices that threshold BQS are the only occurring trust structure. To name some examples,
Aublin \etal~\cite{DBLP:journals/tocs/AublinGKQV15} present an abstraction of an SMR protocol and build BFT algorithms as instances of it.
Liu \etal~\cite{DBLP:conf/opodis/CachinSV16} introduce cross fault-tolerance (XFT), a model that provides guarantees of crash fault-tolerance but tolerates a number of Byzantine faults.
Buchman \etal~\cite{DBLP:journals/corr/abs-1807-04938} present Tendermint, 
a consensus protocol based on the classical PBFT~\cite{DBLP:journals/tocs/CastroL02} algorithm, making use of a novel gossip primitive.
Finally, Yin \etal~introduce HotStuff~\cite{DBLP:conf/podc/YinMRGA19}, a BFT SMR protocol with linear communication complexity.
Threshold BQS have been researched and well understood in practice, but consensus using \emph{generalized BQS} has been unexplored.

\paragraph{Threshold is not enough.}

However, faults and attacks on the nodes in a system often occur in a coordinated way and exhibit substantial dependencies in practice. Using Werner Vogels' words~\cite{Vogels}: ``Many academics will confess to have made the assumption that failures of component are not correlated. This absolutely unrealistic assumption will come back to haunt you in real life, where failures frequently are correlated.''

In this sense, Cachin~\cite{DBLP:conf/dsn/Cachin01} gives an example of a distributed system where the parties are differentiated by location and operating system (OS).
In this scenario, a vulnerability in an OS may result in all parties with that OS being exploited or a hostile action could corrupt all parties in a specific location. 
This example of a generalized BQS explicitly considers correlations and tolerates more failures than possible in the threshold model. 
It highlights the strictly richer trust assumptions we can specify and the resilience we can gain with protocols based on generalized BQS.

As another motivating example from the field of multiparty computation, Gennaro~\cite{DBLP:phd/ndltd/Gennaro96} studies verifiable secret sharing over non-threshold access structures, but using only formulas in disjunctive normal form to build the access structure. As a future direction he calls for a compact representation of an access structure, which would allow any BQS to be expressed, thus leading to more efficient and flexible MPC protocols.
Moreover, Cramer \etal~\cite{DBLP:conf/eurocrypt/CramerDM00} propose MPC protocols over any general trust structure, assuming that the trust structure allows for MPC at all. They work with generalized linear secret sharing scheme, which is analogous to our focus on monotone span programs.

An important tool in encoding generalized BQS are \emph{monotone span programs}~\cite{DBLP:conf/coco/KarchmerW93, DBLP:conf/eurocrypt/Brickell89}. A monotone span program (MSP) is a linear algebraic model of computation, closely related to other models in the theory of computation, such as Boolean formulas and branching programs~\cite{DBLP:conf/coco/KarchmerW93}. 
It is known that monotone span programs are more powerful than monotone circuits. 
Babai \etal~\cite{DBLP:journals/combinatorica/BabaiGW99} prove that there exist functions requiring an exponential-size formula that can be encoded by a linear-size MSP. 
Monotone span programs have also been proved equivalent to linear secret-sharing schemes~\cite{DBLP:conf/coco/KarchmerW93, Beimel96}
and have been used to generalize existing threshold schemes in the fields of secret sharing and multiparty computation (MPC).
Cramer \etal~\cite{DBLP:conf/eurocrypt/CramerDM00} provide constructions for general MPC protocols based on the MSP primitive.

Recent work on consensus protocols has started to consider trust models that generalize the traditional threshold assumption.  Flexible Byzantine fault tolerance~\cite{DBLP:conf/ccs/MalkhiN019}, for instance, considers diverse quorums where some nodes may choose a different threshold quorum.  Asymmetric quorum systems~\cite{DBLP:conf/opodis/CachinT19} let each node specify its own quorum system.

\paragraph{Contributions.}

In this work we focus on generalized BQS and demonstrate the first BFT consensus protocol with generalized quorums. We describe all components necessary for generalized BQS-based protocols and investigate different ways to realize them. In particular, we address all the following topics:

\begin{description}
\item[Encoding a generalized BQS.] We first consider a monotone Boolean formula (MBF)
  consisting of \emph{and}, \emph{or}, and \emph{threshold} operators for
  specifying a BQS.  Since monotone span programs are stronger than
  monotone Boolean formulas, as mentioned, we also investigate MSP for
  representing generalized BQS.  We exhibit an algorithm for turning a BQS
  specification into an MSP. When the BQS is specified as a monotone
  formula, the size of the created MSP is linear in its inputs.

\item[Integrating generalized BQS with consensus.]  For both
  representations (MBF and MSP), we show algorithms for checking quorum
  properties and for integrating them with distributed protocols.
  Comparing the implementations we observe that the MBF-based method
  generally performs better than the MSP-based implementation because of
  the matrix manipulations required by the MSP.  This provides the first
  unified treatment of the efficiencies of these methods and paves the way
  for their practical deployment.

\item[Generalized Byzantine quorum systems.] We apply our methods to
  generalized BQS as described in the literature.  For the M-Grid
  BQS~\cite{DBLP:journals/siamcomp/MalkhiRW00}, which arranges $n$ nodes in
  a square and tolerates $O(\sqrt{n})$ Byzantine nodes, we construct the
  corresponding MSP and investigate its properties. We implement an
  attribute-defined BQS generalizing the OS and location-based example
  mentioned before and represent this as an MBF and as an MSP.

\item[HotStuff consensus with generalized BQS.]  Last but not least, we
  address consensus, the central problem in distributed computing.
  Applying our approach, we realize consensus with generalized BQS by
  building on HotStuff~\cite{DBLP:conf/podc/YinMRGA19}, an efficient BFT
  consensus algorithm.  This is the first BFT consensus implementation
  using a generalized trust assumption.  In benchmarks with up to 40
  replicas, we observe that the performance with the MBF representation is
  comparable to that of the threshold BQS.  Using the same threshold trust
  structure, the MSP representation shows lower performance.
\end{description}

\paragraph{Related work.}
The exploration of generalized structures has a long background in the field of secret sharing.
Benaloh and Leichter~\cite{DBLP:conf/crypto/Leichter88} present the first secret-sharing scheme for arbitrary monotone access structures. They use monotone Boolean formulas with \emph{and}, \emph{or}, and \emph{threshold} operators to express the access structure and introduce a recursive secret-sharing construction. Their scheme is efficient for access structures that can be expressed with polynomially sized formulas.
The MSP model was first used by Brickell~\cite{DBLP:conf/eurocrypt/Brickell89} for secret sharing, although not explicitly identified as such.
After Karchmer and Wigderson~\cite{DBLP:conf/coco/KarchmerW93} formally defined MSP as a model for computation, it has been shown that linear secret-sharing schemes are equivalent to MSP~\cite{Beimel96}.

Many constructions have been suggested for creating the MSP of a given access structure. Lewko and Waters~\cite{DBLP:conf/eurocrypt/LewkoW11a}, in a way similar to Benaloh and Leichter~\cite{DBLP:conf/crypto/Leichter88}, suggested a general algorithm for converting any monotone Boolean formula to an MSP, that is however inefficient for access structures expressed with threshold operators.
The notion of  \emph{insertion} was introduced in by Martin~\cite{Martin93}. Nikov and Nikova~\cite{DBLP:journals/iacr/NikovN04} explored constructions for recursively building the MSP for an access structure from existing MSPs for smaller access structures and presented the definition of insertion used here.

\paragraph{Organization.}
The rest of the paper is organized as follows. Section \ref{sec:basics}
introduces the main concepts and important background. Section \ref{sec:techniques}
presents our techniques for encoding a BQS. In Section \ref{sec:consensus}
we describe HotStuff consensus algorithm with generalized BQS and prove
its consistency and liveness properties.  Section \ref{sec:eval} subsequently
evaluates an implementation of our generalized BQS methods using the HotStuff consensus protocol.

\section{Preliminaries}\label{sec:basics}
\dotparagraph{Parties and failures}
We denote as $\CP = \{p_1, \ldots, p_n \}$ the set of all \emph{parties} in a distributed protocol. 
Whenever describing properties of protocols, we consider Byzantine faults, meaning that faulty parties are allowed to take arbitrary steps, cooperate, and learn the internal state held by any of them. For a specific execution we denote as $B$ the set of the \emph{actually faulty parties}.

\begin{definition}[Fail-prone system~\cite{DBLP:journals/joc/HirtM00}]\label{def:fs}
A \emph{fail-prone system} $\CF \subseteq 2^\CP$ is a set of subsets of $\CP$ such that for every execution there is one \emph{fail-prone set} $F \in\CF$ with $B \subseteq F$. 
A fail-prone system is maximal, in the sense that no fail-prone set contains another one. 
\end{definition}

\begin{definition}[Byzantine quorum system~\cite{DBLP:journals/dc/MalkhiR98}]\label{def:bqs}
Let  $\CF$ be a fail-prone system.
A \emph{Byzantine quorum system (BQS)} $\CQ \subseteq 2^\CP$ is a non-empty set of non-empty subsets of $\CP$,
such that no set is contained in another one,
where each $Q \in \CQ$ is called a \emph{quorum},
satisfying the following properties:
\begin{description}
  \item[Consistency:] \[\forall Q_1, Q_2 \in \CQ , \forall F \in \CF: \, Q_1 \cap Q_2 \not
    \subseteq F.\]
  \item[Availability:] \[\forall F \in \CF: \, \exists Q \in \CQ: \, F \cap Q = \emptyset.\]
  \end{description}
\end{definition}

The definition actually corresponds to a \emph{Byzantine dissemination quorum system}~\cite{DBLP:journals/dc/MalkhiR98}.
When a BQS is defined only by cardinality, i.e., it includes all the subsets of \CP of a given size, it is called a \emph{threshold BQS}. When a BQS is allowed to contain arbitrary subsets of \CP it is called a \emph{generalized BQS}. 

\begin{definition}[$Q^3$-condition~\cite{DBLP:journals/dc/MalkhiR98, DBLP:journals/joc/HirtM00}]\label{def:q3}
  Let \CF be a fail-prone system. We say that \CF satisfies the \emph{$Q^3$-condition} whenever
  \[
    \forall F_1, F_2, F_3 \in \CF: \, \CP \not\subseteq F_1 \cup F_2 \cup F_3.
  \]
\end{definition}
For threshold BQS, the $Q^3$-condition is equivalent to the requirement $n > 3 f$.
Given a fail-prone system \CF, a BQS for \CF exists if and only if \CF satisfies the $Q^3$-condition.
In particular, if $Q^3$ holds, then the bijective complement of the fail-prone sets,
i.e., $\overline{\CF} = \{ \CP \setminus F ~|~ F \in \CF \}$,
is a BQS, called the \emph{canonical BQS} of $\CF$.

\begin{definition}[Access structure~\cite{DBLP:conf/crypto/Leichter88}]
  \label{def:as}
A \emph{monotone access structure} \CA is a collection of non-empty subsets of $\CP$ such that $A \subseteq \CP$ is called \emph{authorized} whenever $A \in \CA$. Monotonicity means that if $A \in \CA$, then any superset $A'$ of $A$ is also in \CA.  The \emph{basis} of \CA is the minimal collection of authorized sets.
\end{definition}

A BQS specifies the quorums that are self-sufficient for a particular task.
The term \emph{access structure} is used more often
in multiparty computation and secret
sharing~\cite{DBLP:conf/crypto/Leichter88, DBLP:journals/joc/HirtM00, DBLP:conf/eurocrypt/CramerDM00}, whereas quorum systems originate in distributed computing~\cite{DBLP:journals/dc/MalkhiR98}. 
Notice that the basis of an access structure is equivalent to our notion of a
(minimal) quorum system in the sense thta every quorum is a (minimal)
authorized set.  In the following, we will use BQS and access
structure interchangeably.  We thus depart slightly from the literature
(and from Definition~\ref{def:as}) by redefining the access structure to its basis and therefore
interpret both as the minimal collection of
subsets of~\CP with a certain property.

The notion of an \emph{insertion}, which we present next, has been introduced as an effort to create
authorized sets  and access structures by combining existing, smaller ones.

\begin{definition}[Insertion on access structures~\cite{Martin93, DBLP:journals/iacr/NikovN04}]\label{def:insertion}
Let $\CA_1$ and $\CA_2$ be two monotone access structures defined on two sets of parties $\CP_1$ and $\CP_2$, respectively, and let $p_z \in \CP_1$ such that $p_z \not\in \CP_2$. The \emph{insertion} of $\CA_2$ at $p_z$, written as $\CA_1(p_z \rightarrow \CA_2)$, is the monotone access structure $\CA_3$ defined on the set $ \CP_3 = (\CP_1 \setminus \{p_z\}) \cup \CP_2$ that satisfies the following: a set $A \subseteq \CP_3$ is authorized in $\CA_3$ if and only if the set $A \cap \CP_1$ is authorized in $\CA_1$ or the set  $A \cap \CP_1$ together with $p_z$ is authorized in $\CA_1$ and $p_z$ is replaced in $A$ by a set authorized in $\CA_2$. Formally,
\begin{equation*}
A \in \CA_3 \ \Leftrightarrow \ \biggl[
\begin{array}{l}
  A \cap \CP_1 \in \CA_1 \: \lor\\
  \left( (A \cap \CP_1)\cup \{p_z\} \in \CA_1 \land  A \cap \CP_2 \in \CA_2 \right)
\end{array}
\biggr].
\end{equation*} 
\end{definition}

\dotparagraph{Monotone span programs~\cite{DBLP:conf/coco/KarchmerW93}}\label{def:msp}
\emph{Monotone span programs} (MSP) have been introduced as a linear-algebraic model of computation.  An MSP is a quadruple $(M, \rho, \evector_1, \CP)$, where $M$ is an $m \times d$ matrix over a finite field \CK, $\rho$ is a surjective function $\{1, \dots, m\} \rightarrow \{p_1, \ldots, p_n \}$ that labels each row of $M$ with a party in \CP, and $\evector_1$ is the vector $(1, 0, \ldots, 0) \in \CK^d$, called the \emph{target vector}.
If $\rvector_i$ is a row of $M$ and $\rho(i) = p_j, p_j \in \CP$, we say that party $p_j$ \emph{owns} row $\rvector_i$. 
There is also a function $\phi: \CP \rightarrow 2^{\{1, \dots, m \}}$, such that $\phi(p_j)$ is the set of rows owned by party $p_j$.
The \emph{size} of the MSP is the number of its rows $m$.

For any set $A \subseteq \CP$ we define $M_A$ to be the $m_A \times d$ matrix 
obtained from $M$ by keeping only the rows $\rvector_i$ with $\rho(i) \in A$,
that is, only the rows owned by parties in $A$.
Let $M_A^\intercal$ denote the transpose of $M_A$ and $\op{Im}(M_A^\intercal)$ the span of the rows of $M_A$. We say that the MSP \emph{accepts} the set $A$ if the rows of $M_A$ span $\evector_1$, that is, $\evector _1\in \op{Im}(M_A^\intercal)$. Equivalently, there is a \emph{recombination vector} $\lvector_A$ such that $\lvector_A M_A = \evector_1$. We say that the MSP \emph{rejects} $A$ otherwise. It follows that each MSP accepts exactly one monotone access structure and each monotone access structure can be expressed in terms of an MSP~\cite{DBLP:conf/coco/KarchmerW93, Beimel96}.

One of the objectives of this work is to construct an MSP that encodes a given BQS, i.e., accepts exactly its quorums. Thus, when working with MSPs (in Section \ref{sec:encode_msp}), we start from a given BQS (and an implicit fail-prone system), such that \emph{consistency} and \emph{availability} of the BQS are satisfied. We usually express this in terms of the access structure equivalent to the BQS.

\section{Techniques}\label{sec:techniques}
When Byzantine quorum systems are allowed to contain arbitrary sets, two questions arise: How will these sets be specified by the user? And how are they encoded within a protocol? A first solution could involve 
an enumeration of all quorums, this would however lead to long user-inputs and large internal representation.
A more efficient solution is hence required, one that provides users with an effective, intuitive and user-friendly way to specify a BQS.
It is also crucial to internally encode the BQS using a data structure that is efficient, able to encode any possible BQS, and also offering an inexpensive method for checking whether a set is a quorum. 

\subsection{Generalized Byzantine quorum systems as formulas}
In this section we show how the generalized trust assumptions of the system can be specified by the user in a structured way and encoded within the protocol as a Boolean formula.

We observe that it is enough to use only the \emph{threshold operator} $\Theta_k^m(q_1, \ldots, q_m)$,
which specifies that any subset of $\lbrace q_1, \ldots, q_m \rbrace$ with cardinality $k$ is a quorum.
Each $q_i$ can be a literal, i.e., a party identifier, or a nested threshold operator.
The threshold operator is the generalization of logical \emph{conjunction},
that would require all $q_i$s to make a quorum, 
and logical \emph{disjunction}, that would allow each of them alone to be a quorum 
-- the first can be obtained for $k = m$ and the second for $k = 1$.
The threshold operator is thus complete, in the sense that it can describe any possible BQS.
Therefore, the users are allowed to specify the generalized trust assumptions in a standard format like JSON,
using nested threshold operators.
This is aligned with the way users specify their quorum slices in Stellar Blockchain~\cite{Mazieres16} with threshold operators.

We use the notion of a \emph{monotone Boolean formula} (MBF), a formula that consists of \emph{and}, \emph{or}, and \emph{threshold} operators and literals that correspond to parties.
An MBF $F$ describes a monotone function $2^\CP \rightarrow \{0, 1\}$ in the following way;
when $F$ consists only of a literal, then the value of $F$ on input $S \subseteq \CP$ is 1 if and only if $F \in S$;
when $F$ is the threshold operator $\Theta_k^m(q_1, \ldots, q_m)$, then $F(S)$ is 1 if at least $k$ of the  $q_1, \ldots, q_m$ are recursively evaluated to 1 on input $S$;
and accordingly for the other operators.
We say that an MBF $F$ \emph{implements} a BQS \CQ if it returns 1 on input a set $A \supseteq Q$, for $Q \in \CQ$, and 0 otherwise. 

We use a tree data structure to store a BQS described through an MBF, where the internal nodes represent an operator, their children are the operands, and the leaves always represent a party.
Clearly, the size of the tree (defined as the number of nodes) is linear in the quorum specification given by the user.
We employ Algorithm \ref{alg:isAuthorized_formula} to evaluate whether a set is considered a quorum in the BQS implemented by a formula $F$. The runtime is linear in the size of $F$, given that the set membership operation returns in constant time.

\begin{algo*}
\vbox{
\small
\begin{numbertabbing}\reset
xxxx\=xxxx\=xxxx\=xxxx\=xxxx\=xxxx\=MMMMMMMMMMMMMMMMMMM\=\kill
    \textbf{eval}$(F, A)$ \label{}\\ 
	\> \textbf{if} \( F \) is a literal \textbf{ then } \label{}\\
	\>\> \textbf{return} \( (F \in A) \) \label{}\\
	\> \textbf{else}  \label{}\\
	\>\> write $F = \op{op}(F_1, \ldots, F_m)$, where
             $\op{op} \in \{\land, \lor, \Theta \}$ \label{}\\	
    \>\> \textbf{for each} \( F_i \) \textbf{do} \label{}\\
	\>\>\> \(x_i \gets \textbf{eval}(F_i, A) \) \label{}\\
	\>\> \textbf{return} \( \op{op}(x_1, \ldots, x_m) \) \label{}
 	\end{numbertabbing}
}
\caption{Checking whether set $A$ is a quorum in the BQS implemented by formula $F$.}
\label{alg:isAuthorized_formula}
\end{algo*}

\dotparagraph{A layered BQS}\label{sec:one_com}
An example that highlights a more complex BQS that cannot be specified in the threshold model is a \emph{2-layered-1-common} BQS (\onecom). This example shows a hierarchical trust structure with a notion of proximity that models a realistic system structured into two levels. To our knowledge, it has not been used in practice so far.
Let us consider two disjoint sets of parties, organized in two layers, with $k$ parties $A_0 \ldots A_{k-1}$ on the first 
and $3k$ parties $B_0 \ldots B_{3k-1}$ on the second.
We may assume that the parties in the first layer are more trusted than
those in the second layer. A quorum consists of a strict $2/3$ majority of the parties in the first layer 
plus, for each party $A_\ell$ of these, a 2 out-of 4 threshold from the set $\{B_{3\ell}, \ldots, B_{3(\ell + 1)} \} $, where indices are modulo $3k$.
For $k \in \BN$, the general formula of the BQS is 
\begin{equation}\label{eq:1common}
\Theta_{\ceil{\frac{2k + 1}{3}}}^k \left(\lbrace  A_\ell 
\land \Theta_2^4 \left(\lbrace B_{m}  \rbrace\right)  \rbrace\right), 
\text{ for } \ell \in \{0, \ldots, k-1\} \text{ and } m \in \{3\ell, \ldots, 3(\ell+1) \op{ mod } 3k\}.
\end{equation}

A 2-layered-1-common BQS for $k = 4$ can be seen in Figure \ref{fig:1common}. In Figure \ref{fig:configjson} we show a configuration file that specifies this BQS~-- it is actually the file used during the evaluation. It is worth to notice that this BQS, even for $k = 4$, results in a system with 792 quorums, which highlights why a naive, quorum-enumeration solution would be impractical.
Notice that by using the fail-prone system that corresponds to a \onecom~BQS in the canonical way, we observe that this BQS satisfies the $Q^3$-condition because every fail-prone set contains fewer than $k / 3$ parties from the first layer. 
Thus, it is indeed a BQS.

\begin{figure}
\begin{center}
	\includegraphics[width=\textwidth]{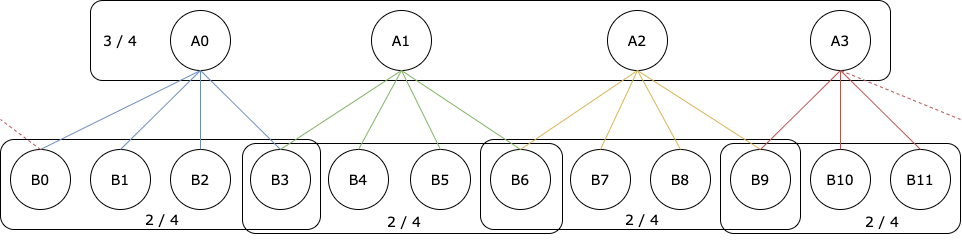}
\end{center}
\caption{The \onecom~Byzantine quorum system for $k = 4$.
The corresponding MBF is
$\Theta_3^4(A_0 \land \Theta_2^4(B_0, B_1, B_2, B_3), A_1 \land \Theta_2^4(B_3, B_4, B_5, B_6), A_2 \land \Theta_2^4( B_6, B_7, B_8, B_{9}), A_3 \land \Theta_2^4( B_{9}, B_{10}, B_{11}, B_0)  ) .
$}\label{fig:1common}
\end{figure}

\begin{figure}
\begin{center} 	
	\includegraphics[width=\textwidth]{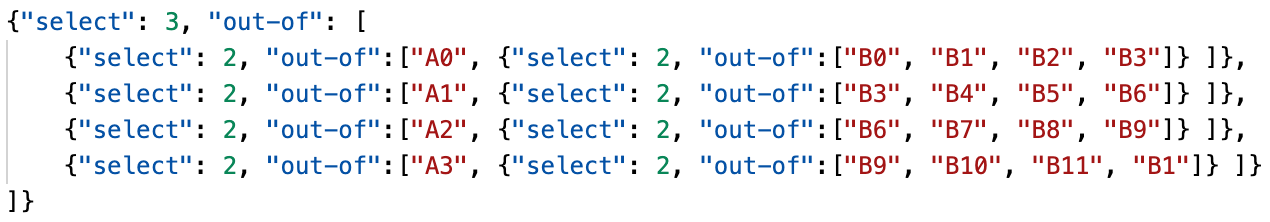}
\end{center}
\vspace*{-3ex}\caption{A specification of the \onecom~Byzantine quorum system for $k = 4$ in JSON format using nested threshold operators.}\label{fig:configjson}
\end{figure}

\subsection{Generalized Byzantine quorum systems as monotone span programs}\label{sec:encode_msp}
Until now, we considered BQS that can be efficiently encoded using formulas.
However, as already discussed, results in complexity theory suggest that MSPs can be superpolynomially stronger than monotone formulas.
Moreover, the MSP is a compact and concise data structure, that can be encoded by a matrix and a vector over a field.
For these reasons, we also investigate the capabilities of the MSP as the data structure that encodes a BQS. 
In this section we show how to instantiate an MSP from an MBF and how the MSP can be used to check for quorums.
Later, we evaluate the MSP-based implementation and compare it with the one based on MBF.
We remark, however, that constructing the MSP from an MBF is not the only option; in case a BQS is more efficiently described by an MSP than by a formula,
we could plug the MSP directly in the protocol and use the same quorum-checking algorithms.
Throughout this section, we formulate all our results in terms of the access structure implied by the given BQS, 
since we only focus on the quorums of the BQS and not its other properties.

In line with our previous terminology, we say that an MSP \CM \emph{implements} an access structure \CA if it accepts exactly the sets in \CA and their supersets.
Returning to the idea of \emph{insertion}, we first show how this notion is reflected on the MSP that implements an access structure.
In the following, let $\CM^{(k)} = (M^{(k)}, \rho^{(k)}, \evector_1^{(k)}, \CP^{(k)} )$ be MSPs, where $M^{(k)}$ has dimensions $m_k \times d_k$, for $k \in \{1, 2, 3\}$.
We denote the rows of each $M^{(k)}$ as $\rvector_i^{(k)}$, for $1 \leq i \leq m_k$.
We also denote the $j^{\text{th}}$ column in a row $\rvector$ as $\rvector[j]$, a range of columns $j_1$ to $j_2$ as $\rvector[j_1 : j_2]$, a row with $\ell$ zero elements as $\zerovector^\ell$, and the concatenation of two rows $\rvector$ and $\rvector^\prime$, that is a new vector of size $|\rvector| + |\rvector^\prime|$, whose first elements are $\rvector$ and the last are $\rvector^\prime$, as $\rvector \mid\mid \rvector^\prime$.

\begin{definition}[Insertion on MSPs~\cite{DBLP:journals/iacr/NikovN04}]\label{def:insertion_msp}
Let $\rvector_z$ be a row of $M^{(1)}$ owned by $p_z \in \CP^{(1)}$ -- assuming without loss of generality it is unique.
The \emph{insertion} of $M^{(2)}$ in row $\rvector_z$ of $M^{(1)}$, written as $\CM^{(1)}(\rvector_z \rightarrow \CM^{(2)})$, is an MSP $\CM^{(3)}$, where $M^{(3)}$ has rows identical to $M^{(1)}$, except for $\rvector_z$, which is repeated $m_2$ times in $M^{(3)}$, each time multiplied by the first column of $M^{(2)}$ and with the rest of the columns 2 to $d_2$ of $M^{(2)}$ appended in the end. The function $\rho^{(3)}$ labels the rows of $M^{(3)}$ with the same owners as $\rho^{(1)}$, except for $\rvector_z$. The newly inserted rows are labeled according to $\rho^{(2)}$.

More formally, $M^{(3)}$ is an $(m_1 + m_2 - 1) \times (d_1 + d_2 - 1)$ matrix with rows
\begin{equation}\label{eq:M3_M}
\rvector_i^{(3)} = 
\begin{cases} 
      \rvector_i^{(1)} \mid\mid \zerovector^{d_2 - 1} 									& 1 \leq i \leq z - 1\\
      \rvector_z * \rvector_{i-z+1}^{(2)}[1] \mid\mid \rvector_{i-z+1}^{(2)}[2 : d_2] 	& z \leq i \leq z + m_2 - 1 \\
      \rvector_{i-m_2+1}^{(1)}  \mid\mid \zerovector^{d_2 - 1}	  				 	& z + m_2 \leq i \leq m_1 + m_2 - 1
   \end{cases}
\end{equation}
and $\rho^{(3)}$ is a surjective function $\{1, \dots, m_1 + m_2 - 1\} \rightarrow (\CP^{(1)} \setminus \{p_z\}) \cup \CP^{(2)}$ defined as
 \begin{equation*}
\rho^{(3)}(i) = 
\begin{cases} 
      \rho^{(1)}(i)						& 1 \leq i \leq z - 1\\
      \rho^{(2)}(i - z + 1)		 		& z \leq i \leq z + m_2 - 1 \\
      \rho^{(1)}(i - m_2 + 1 )	  	& z + m_2 \leq i \leq m_1 + m_2 - 1
   \end{cases}
\end{equation*}
\end{definition}

\begin{lemma}~\cite{DBLP:journals/iacr/NikovN04}\label{lem:insertion}
If an MSP $\CM^{(1)}$ implements the access structure $\CA^{(1)}$, with row $\rvector_z$  owned by party $p_z$, and an MSP $\CM^{(2)}$ implements the access structure $\CA^{(2)}$, then the MSP $\CM^{(1)}(\rvector_z \rightarrow \CM^{(2)})$ implements the access structure $\CA^{(1)}(p_z \rightarrow \CA^{(2)})$.
\end{lemma}

\begin{lemma}\label{lem:vandemonde}
Let \emph{Vandermonde-MSP}$(n, t, \CP)$ be defined as the MSP $\left( V(n,t), \rho, \evector_1, \CP \right)$, with
$\CP = \{p_1, \ldots, p_n \}$,
$V(n,t)$ the $n \times t$ Vandermonde matrix over a finite field $\CK$,
\begin{equation*}
V(n,t) = 
\begin{pmatrix}
1  &  x_1  &  x_1^2  &  \cdots  &  x_1^{t-1} \\
1  &  x_2  &  x_2^2  &  \cdots  &  x_2^{t-1} \\
\vdots  &  \vdots  &  \vdots  &  \ddots  &  \vdots \\
1  &  x_n  &  x_n^2  &  \cdots  &  x_n^{t-1} \\
\end{pmatrix},
\end{equation*}
with $x_i \neq x_j \neq 0$, for $1 \leq i \leq j \leq n$,
$\rho$ a function that maps row $\rvector_i$ to party $p_i$, for $i \in \{1, \ldots, n\}$,
and $\evector_1 = (1, 0, \ldots, 0) \in \CK^t$.
Then, Vandermonde-MSP$(n, t, \CP)$ implements the $t$ out-of $n$ threshold access structure $\Theta_t^n(\CP)$.
\end{lemma}

\begin{proof}
Let $A \subset \CP$ and $M_A$ the matrix consisting of the rows of $M$ owned by the members of $A$. 
From the results of linear algebra, and because $x_i$'s are pairwise different, we know that the rank of $M_A$ is maximal (that is, equal to $t$, and thus $\op{Im}(M_A^\intercal) = \CK^t$) if and only if $|A| \geq t$.
Therefore, $\CM$ accepts exactly those sets $A$ with $|A| \geq t$.
\end{proof}

\dotparagraph{Building the MSP that implements a generalized BQS}
Based on the previous lemmata, we now present Algorithm~\ref{alg:insertion_msp} that gets as input an access structure, encoded as an MBF, and outputs an MSP that implements it. 
The idea is to start with a Vandermonde matrix implementing the first in the hierarchy threshold operator and repeatedly perform insertions of the MSPs implementing the nested threshold operators.

Algorithm~\ref{alg:insertion_msp} works as follows.
Let $F = \Theta_d^m(F_1, \dots, F_m)$ be an MBF, where each $F_i$ can be a party or a nested threshold operator. 
The algorithm first creates the MSP for $F$ (lines \ref{line:one}--\ref{line:vander}) in the following way:
it extracts the values $m, d$ and $F_1, \dots, F_m$ from $F$ (line \ref{line:extract})
and examines whether each $F_i$ is a party literal or a nested operator.
In the second case, a fresh \emph{virtual party} $v_i$ is created and associated with $F_i$
(the map $V_{\text{map}}$ is used to keep track of this association).
A virtual party is treated exactly as an actual party, except it is used only during this construction.
The MSP for $F$ is now created according to Lemma \ref{lem:vandemonde} and using
both actual and virtual parties as the set \CP.
In the second part of the algorithm (lines \ref{line:secondPart_start}--\ref{line:secondPart_end})
the MSPs for the nested operators (virtual parties $v_i$) are recursively created (line \ref{line:recurse}) 
and inserted in \CM, according to Definition \ref{def:insertion_msp}.
The function $\phi$ related to the MSP $\CM$, that maps a party to the rows they own, 
is used to get the row $\rvector_i$ of $M$ that was labeled with $v_i$.
Notice that in line \ref{line:fresh_vi}, a fresh variable is created 
for each nested operator, so $v_i$ owns a single row.

For the termination of the recursion, notice that, if $F$ does not contain
any nested threshold operators, $V$ is the empty set when we reach line
\ref{line:secondPart_start}, and the algorithm returns.  The next result
therefore follows immediately from the definition of \emph{insertion} and
the fact that the algorithm starts with a $1 \times 1$ matrix.

\begin{algo*}
\vbox{
\small
\begin{numbertabbing}\reset
xxxx\=xxxx\=xxxx\=xxxx\=xxxx\=xxxx\=MMMMMMMMMMMMMMMMMMM\=\kill
	\textbf{buildMSP}$(F)$ \label{line:one}\\
	\> \( \text{let } \Theta_d^m(F_1, \ldots, F_m ) \text{ be the formula } F \) \label{line:extract}\\	
	\> \(R \gets \emptyset \) \label{}\\
	\> \(V \gets \emptyset \) \label{}\\
	\> \(V_{\text{map}} \gets \emptyset \) \label{}\\
	\> \textbf{for each} \( F_i \) \textbf{do}\label{}\\
	\>\> \textbf{if} \( F_i \) is a literal $p$ \textbf{then} \label{}\\
	\>\>\> \( R \gets R \cup \{p\} \) \label{}\\
	\>\> \textbf{else}  \label{}\\
	\>\>\> declare \( v_i \) a new virtual party  \label{line:fresh_vi}\\
	\>\>\> \( V \gets V \cup \{ v_i \} \) \label{}\\
	\>\>\> \( V_{\text{map}} \gets V_{\text{map}} \cup \{ (v_i, F_i) \} \) \label{line:create_qi}\\
	\> \( \CM \gets \op{Vandermonde-MSP}( m, d,  R \cup V) \) \label{line:vander} \\
	\> \textbf{for each} \( v_i \in V \) \textbf{do}\label{line:secondPart_start}\\
	\>\> $\CM_2 \gets$ \textbf{buildMSP}$\left( V_{\text{map}}(v_i) \right) $ \label{line:recurse}\\
	\>\> \( r_i \gets \phi(v_i) \) \label{}\\
	\>\> \( \CM \gets \CM(r_i \rightarrow \CM_2) \) \label{line:secondPart_end}\\
	\> \textbf{return} \( \CM \) \label{}
\end{numbertabbing}
}
\caption{Construction of an MSP from a monotone Boolean formula $F$.}
\label{alg:insertion_msp}
\end{algo*}

\begin{lemma}\label{lem:msp_size}
Let $F$ be an MBF that includes in total $c$ operators in the form $\Theta_{d_i}^{m_i}$. The matrix $M$ of the MSP constructed with Algorithm~\ref{alg:insertion_msp} has $m = \sum_1^k m_i - c + 1$ rows and $d = \sum_1^k d_i - c + 1$ columns. 
\end{lemma}

Lemma~\ref{lem:msp_size} implies that the resulting matrix $M$ has size linear in the length of~$F$. 
In the special case that each party appears only once in the access structure, $M$ has $n$ rows and at most $n$ columns, where $n = | \CP|$.

\dotparagraph{Checking for quorums}
We now show how to determine whether a set constitutes a quorum using the MSP representation of the system and no other information about the BQS (e.g., whether it is a threshold or a generalized BQS, or whether it was specified using threshold or other operators).

We have seen that an MSP accepts a set $A$ if and only if the rows of $M_A$ span the vector $\evector$,
or, equivalently, the linear system $M_A^T \xvector = \evector$ has solutions for $\xvector$. 
According to linear algebra, a necessary and sufficient condition for this is that  %Rouché–Capelli theorem
the rank of $M_A^T$ is equal to the rank of the augmented matrix $M_A^T|\evector$. To check this condition, we perform Gaussian elimination on the augmented matrix $M_A^T|\evector$ and bring it in row echelon form. If it contains a row with only zeros in the coefficient part but a nonzero value in corresponding constant part, then the rank of $M_A^T|\evector$ is bigger that the rank of $M_A^T$ and $A$ is not an authorized group, otherwise, $A$ is authorized.

Gaussian elimination has a cubic time complexity, so it is expensive to perform it every time we wish to check for a quorum. As an optimization we use the LUP-decomposition of 
matrix $M^T$, i.e., we calculate the $d \times d$ matrices $P$ and $L$, and the $d \times m$ matrix $U$, such that $P M^T =L U$. Then, for any set $A$ we get $P M_A^T =L U_A$, where $P$ and $L$ do not depend on $A$. In the initialization of the protocol we solve $L \yvector = P \evector$ for $\yvector$, where $\yvector$ is a $d$-vector. Then, instead of the equation $M_A^T \xvector = \evector$ we can work with the equation $U_A \xvector = y$. In order to check whether a set $A$ is authorized, we now have to bring $U_A|\yvector$ in row echelon form. Since $U_A$ is an upper triangular matrix, some computational steps are avoided.

Notice here that it might be the case that $A$ is a superset of an authorized group. These redundant parties can easily be identified from the echelon form, as they will correspond to the free variables of the system -- variables whose corresponding column does not contain a pivot. Another situation worth to mention is that a party can own more than one rows of $M$. However, the algorithm described above also works in this case, since $M_A$ will contain all rows owned by parties in $A$.

\subsection{Concrete constructions of Byzantine quorum systems}
We now consider two specific families of generalized BQS that have been studied in the literature and show how they can be encoded as MSPs.

\dotparagraph{Attribute-based BQS}
A BQS of this family is defined over a set of attributes, which are associated with the parties, and a quorum is described in terms of required attributes.
Let $\CX = \{ \chi_1, \ldots, \chi _r\}$  denote the set of attributes and $\Psi \subseteq \CP \times \CX$ the relation between parties and attributes. We say that party $p_j$ \emph{holds} an attribute $\chi$ whenever $(p_j , \chi) \in \Psi$.
An \emph{attribute-based MBF} is a monotone Boolean formula $F(\chi_1, \ldots, \chi_r)$ over the attributes $\CX$
and implements a BQS where a set $A \subseteq \CP$ is a quorum 
whenever the attribute set $\{ \chi \in \CX \mid \exists p \in A: (p, \chi) \in \Psi \}$, collectively held by the parties in $A$, satisfies~$F$.
By adding one more syntactic rule, we can also specify the requirement that an attribute is held by at least a number of parties.
Let each $\chi_i \in \CX$ be related with $L_i$ parties, 
i.e., $\{ p \in \CP \mid (p, \chi_i) \in \Psi \} = L_i$, and let $\ell_i \leq L_i$.
Then, a formula $F(\chi_1^{(\ell_1)}, \ldots, \chi_r^{(\ell_r)})$
specifies that $A$ is a quorum if, in addition to the aforementioned condition,
each $\chi_i$ is held by at least $l_i$ parties, i.e.,
$| \{ p \in A \mid (p, \chi_i) \in \Psi \} | \geq \ell_i$, for $ 1 \leq i \leq r$.

An MSP $\CM = (M, \rho, \evector_1, \CP)$ that implements $F(\chi_1^{(\ell_1)}, \ldots, \chi_r^{(\ell_r)} )$ can be constructed as follows. 
First, an MSP $\CM^\prime = (M^\prime, \rho^\prime, \evector_1^\prime, \CF)$ is created for $F(\chi_1, \ldots, \chi_r)$,
using the methods presented in the previous sections.
Then an insertion $\CM^\prime(r_i \rightarrow \CM_i)$ is performed for every $\chi_i$,
as described in Definition~\ref{def:insertion_msp},
where $\CM_i$ is an MSP such that $M_i$ is the $L_i \times \ell_i$ Vandermonde matrix
and $\rho_i$ is a function labelling the rows of $M_i$ with the parties related to $\chi_i$.
Notice that the resulting MSP $\CM$ is defined on the set of parties \CP and not the set of attributes \CX.

We now instantiate the attribute-based BQS mentioned in the introduction using this methodology.
Recall that there are two families of attributes, location and operating system. 
We use the attributes $\{ \chi_{11}, \chi_{12}, \chi_{13}, \chi_{14}\}$  for the four different locations and 
the attributes $\{ \chi_{21}, \chi_{22}, \chi_{23}, \chi_{24}\}$  for the four different OS.
The 16 parties are arranged in a four by four grid, so that each party is related with exactly one attribute from each family.
The system tolerates the simultaneous failure of all parties in one location and all parties with a specific OS. 
Thus, a set is a quorum if it contains at least three parties with different OS for at least three different locations.
This BQS is implemented by the attribute-based MBF 
\begin{equation*}
\Theta_{3}^{4}\left( \chi_{11}^{(3)}, \chi_{12}^{(3)}, \chi_{13}^{(3)}, \chi_{14}^{(3)} \right) 
\land \Theta_{3}^{4}\left( \chi_{21}^{(3)}, \chi_{22}^{(3)}, \chi_{23}^{(3)}, \chi_{24}^{(3)} \right) .
\end{equation*}
Following the method described above, a $4 \times 3$ Vandermonde matrix will be inserted in every $\chi_{ij}^{(3)}$ when creating the MSP, 
which, according to Lemma \ref{lem:msp_size}, will have dimensions $32 \times 22$.

\dotparagraph{The M-Grid BQS}
Malkhi \etal~\cite{DBLP:journals/siamcomp/MalkhiRW00} proposed the \emph{M-Grid} system, 
a family of BQS where
$n = k^2$ parties are arranged in a $k \times k$ grid 
and up to $b$ parties are allowed to be Byzantine, with $b \leq (\sqrt{n+1}) / 2$.
A quorum consists of any $\sqrt{b+1}$ rows and $\sqrt{b+1}$ columns.
Actually, the M-Grid was proposed as a \emph{Byzantine masking quorum system}~\cite{DBLP:journals/dc/MalkhiR98}, 
a category of BQS that requires a stronger intersection property than the
Byzantine dissemination quorum systems, but one can adapt the construction accordingly.

For a dissemination BQS, the requirement for $b$ is $b \leq \sqrt{n} - 1$ 
and a quorum consists of any $\sqrt{b/2+1}$ rows and $\sqrt{b/2+1}$ columns.
To see this, notice that if two quorums $Q_1$ and $Q_2$ have a row or a column in
common, then $|Q_1 \cap Q_2| \geq \sqrt{n} \geq b + 1$. 
Otherwise, the intersection of $Q_1$'s columns with $Q_2$'s rows is disjoint from the 
intersection of $Q_2$'s columns with $Q_1$'s rows, so $|Q_1 \cap Q_2| \geq 2\sqrt{b/2+1}\sqrt{b/2+1} > b + 1$.
In both cases, the \emph{consistency} property of a BQS is satisfied.

To encode the M-Grid BQS
we define the attribute set $\CX = \{ R_1, \ldots	, R_{k}, C_1, \ldots, C_{k} \}$
and assign the party $s_{ij}$ at row $i$ and column $j$ the attributes $R_i$ and $C_j$.
The attribute-based MBF related to this BQS family is
\begin{equation*}
\Theta_{\sqrt{b/2+1}}^{k}\left( R_1^{(k)}, \ldots, R_{k}^{(k)} \right) 
\land \Theta_{\sqrt{b/2+1}}^{k}\left(  C_1^{(k)}, \ldots, C_{k}^{(k)} \right).
\end{equation*}
The formula has $3 + 2k$ threshold operators, considering the \emph{and} operator as a 2-out-of-2 threshold and
recalling that our method inserts a $k \times k$ MSP in the attributes $R_i^{(k)}$ and $C_j^{(k)}$.
The resulting MSP that implements the M-Grid BQS has
$2n$ rows and 
$2n + 2( \sqrt{b/2+1} - k ) < 2n$ columns, by Lemma~\ref{lem:msp_size}.

\section{Consensus using generalized quorums systems}\label{sec:consensus}

HotStuff~\cite{DBLP:conf/podc/YinMRGA19} is an efficient leader-driven Byzantine fault-tolerant state-machine replication (SMR) algorithm. The nodes that take part in the protocol are separated into \emph{replicas}, which actually run the protocol, and \emph{clients}, which submit requests to the replicas and receive totally-ordered responses. The trust assumptions are specified by the number of replicas $n$ and the number of tolerated faults $f$. The replicas maintain a tree structure, whose \emph{nodes} contain batches of clients' commands and get committed in a monotonically increasing way. Two nodes \emph{conflict} if none of them extends from the other. 

HotStuff is presented in three versions, the so-called \emph{basic}, \emph{chained}, and \emph{implemented}. In the \emph{basic} version, each view consists of four phases, called \emph{prepare}, \emph{pre-commit}, \emph{commit}, and \emph{decide}. In each phase, the leader waits for $n - f$ different vote messages from the replicas, constructs a \emph{quorum certificate} (QC) upon receiving them, and starts the next phase by broadcasting this certificate to the replicas. The view changes in the end of the \emph{decide} phase, or whenever the replicas time out waiting for a leader's message. Each view has a deterministically determined leader.
The \emph{chained} version pipelines the four phases into one \emph{generic} phase. This serves as the prepare phase for the new node in the tree, as the pre-commit phase for the previous node and so on, so that the four phases map to four successively ordered requests.
Finally, the implemented version presents further optimizations.
The prototype implementation of threshold HotStuff, which we also use for generalized HotStuff, is based on the \emph{implemented} version.

The generalized HotStuff protocol is \emph{instantiated} with a Byzantine quorum system, which specifies its trust assumptions. 
The leader now collects votes from a quorum of processes and constructs a QC by concatenating them. Upon receiving the QC, the replicas
validate the signatures, as well as the fact that the voters indeed form a quorum.	A quorum of processes is also required to trigger a view change
against a faulty leader.

The pseudocode of basic HotStuff with generalized BQS is presented in Algorithm \ref{alg:basic-hotstuff}. We give a brief description of the data structures used and refer to~\cite{DBLP:conf/podc/YinMRGA19} for more details.
A message consists of four fields, \type, \viewNumber, \node, and \justify. 
The \type can be one of \str{new-view}, \str{prepare}, \str{pre-commit}, \str{commit}, \str{decide}. 
The \viewNumber is always populated with the current view number. 
The field \node is used in the \emph{prepare} phase by the leader to propose the new leaf node, as well as by replicas in vote messages. 
Finally, \justify is always used by the leader to send a valid QC and by the replicas to send their \prepareQC in a \str{new-view} message. 
A vote message, sent by replicas, additionally contains a signature over the fields \type, \viewNumber, \node.
The QC data structure consists of four fields, \type, \viewNumber, \node, and \sig. 
The \type can be one of  \str{prepare}, \str{pre-commit}, \str{commit} and is used to indicate the phase in which the votes used to construct the QC were cast. %In the chained version there is a single \str{generic} type. 
The fields \viewNumber and \node indicate the view in which the QC was created and the node it justifies, respectively. 
Finally, the field \sig contains the signatures on the vote messages of the quorum that was used to construct the QC.

In the pseudocode we omit the details related to the signing and verification of the messages, the verification of a QC and the signing of the vote messages. 
We denote as $p_\ell$ the leader of a view. As in the original protocol, this could be 
any deterministic function from the view number to the replicas, as long as it eventually proposes a correct leader. 
If an interrupt happens when replicas are waiting for a message, line~\ref{line:new-view-interrupt} is executed.
The variables \newview, \prepares, \precommits, and \commits, used by the leader to store the votes until a quorum
is received, are emptied in each view (not shown for brevity).

Hotstuff works in the partial-synchrony model~\cite{DBLP:journals/jacm/DworkLS88},
where there is 
an unknown \emph{Global Stabilization Time (GST)}, after which
the communication between two correct replicas becomes synchronous.
The safety of the HotStuff protocol as presented in~\cite{DBLP:conf/podc/YinMRGA19} is based on the properties of threshold Byzantine quorum systems, namely the $n > 3f$ condition. 
In the generalized protocol the safety is reduced to the properties of the generalized BQS.
The generalized version of HotStuff satisfies the same safety and liveness theorems as threshold HotStuff, which we now present and prove for the generalized case.

\begin{algo*}
\vbox{
\small
\begin{numbertabbing}\reset
 xxxx\=xxxx\=xxxx\=xxxx\=xxxx\=xxxx\=MMMMMMMMMMMMMMMMMMM\=\kill
 \textbf{State} \quad \(\prepareQC \gets \perp\); \(\lockedQC \gets \perp\); \(\curView \gets 1 \) \\
  
  // \textit{PREPARE} phase \\
 \textbf{upon} receiving a message \( \msg{new-view}{\viewNumber, \node, \justify} \) from $p_j$ \` 		// only leader $p_{\ell}$ \label{}\\
 \>\>\>\textbf{such that} $ \viewNumber = \curView - 1 $ \textbf{do} \label{}\\ 
 \> \( \newview[j] \gets \justify \) \label{}\\
 \>\textbf{if  exists} \( \{ p_k \in \CP ~|~ \newview[k] \neq \perp \} \in \CQ \) \textbf{then} \label{} \\
 \>\> \( V =\lbrace \newview[k] ~|~ \newview[k] \neq \perp \rbrace \); \( \highQC \gets \text{argmax}_{v \in V }(v.\viewNumber) \)\label{}\\
 \>\> \( \var{curProposal} \gets \text{new node} \) \label{}\\
 \>\> \( \var{curProposal}.\var{parent} \gets \highQC.\node \);  \( \var{curProposal}.\var{cmd} \gets \text{client's command}  \)  \label{}\\
 \>\> send message \( \msg{\str{prepare}}{\curView, \var{curProposal}, \highQC} \) to all $p_{j} \in \CP$ \label{} \\
 \\
 \textbf{upon} receiving a message \( \msg{prepare}{\viewNumber, \node, \justify} \)  from $p_\ell$ \textbf{such that} $ \viewNumber = \curView $ \textbf{do}\label{}\\
 \>\textbf{if} \( \node \text{ extends from \justify.\node} \) \label{}\\
 \>\>\>\textbf{and} \( ( \node \text{ extends from } \lockedQC.\node  \) \label{line:safety}\\
 \>\>\>\textbf{or} \(  \justify.\viewNumber > \lockedQC.\viewNumber ) \) \textbf{then}\label{line:liveness}\\
 \>\> send vote message \(  \msg{\str{prepare}}{\curView, \node, \perp}  \) to $p_{\ell}$  \label{}\\
 \\ 
 // \textit{PRE-COMMIT} phase \\
 \textbf{upon} receiving a vote message \( v = \msg{prepare}{\viewNumber, \node, \justify} \) from $p_j$ \` 		// only leader $p_{\ell}$ \label{}\\
 \>\>\>\textbf{such that} $ \viewNumber = \curView$ \textbf{do} \label{}\\ 
 \> \( \prepares[j] \gets v \) \label{}\\
 \>\textbf{if  exists} \( \{ p_k \in \CP ~|~ \prepares[k] \neq \perp \} \in \CQ \) \textbf{then} \label{} \\
 \>\> \( V =\lbrace \prepares[k] ~|~ \prepares[k] \neq \perp \rbrace \); \( \prepareQC \gets \op{QC} \left( V \right)  \) \label{} \\
 \>\> send message \( \msg{\str{pre-commit}}{\curView, \perp, \prepareQC} \) to all $p_{j} \in \CP$  \label{}\\
 \\
 \textbf{upon} receiving a message \( \msg{pre-commit}{\viewNumber, \node, \justify} \) from $p_\ell$ \label{}\\
 \>\>\>\textbf{such that} $ \viewNumber = \curView $ \textbf{and} $\justify.type = \str{prepare}$ \textbf{do} \label{}\\
 \> \( \prepareQC \gets \justify  \) \label{} \\
 \> send vote message \( \msg{\str{pre-commit}}{\curView, \justify.\node, \perp} \) to $p_\ell$  \label{}\\
 \\ 
// \textit{COMMIT} phase \\
 \textbf{upon} receiving a vote message \( v = \msg{pre-commit}{\viewNumber, \node, \justify} \) from $p_j$ \` 		// only leader $p_{\ell}$ \label{}\\
 \>\>\>\textbf{such that} $ \viewNumber = \curView$ \textbf{do}\label{} \\ 
 \> \( \precommits[j] \gets v \) \label{}\\
 \>\textbf{if  exists} \( \{ p_k \in \CP ~|~ \precommits[k] \neq \perp \} \in \CQ \) \textbf{then} \label{} \\
 \>\> \( V =\lbrace \precommits[k] ~|~ \precommits[k] \neq \perp \rbrace \); \( \precommitQC \gets \op{QC} \left( V \right)  \) \label{} \\
 \>\> send message \( \msg{\str{commit}}{\curView, \perp, \precommitQC} \) to all $p_{j} \in \CP$ \label{} \\
 \\
 \textbf{upon} receiving a message \( \msg{commit}{\viewNumber, \node, \justify} \) from $p_\ell$ \label{} \\
 \>\>\>\textbf{such that} $ \viewNumber = \curView $ \textbf{and} $\justify.type = \str{pre-commit}$ \textbf{do} \label{}\\
 \> \( \lockedQC \gets \justify  \)  \label{}\\
 \> send vote message \(  \msg{\str{commit}}{\curView, \justify.\node, \perp} \) to $p_\ell$ \label{} \\
 \\
 // \textit{DECIDE} phase \\
 \textbf{upon} receiving a vote message \( v = \msg{commit}{\viewNumber, \node, \justify} \) from $p_j$ \` 		// only leader $p_{\ell}$ \label{}\\
 \>\>\>\textbf{such that} $ \viewNumber = \curView$ \textbf{do} \label{}\\ 
 \> \( \commits[j] \gets v \) \label{}\\
 \>\textbf{if  exists} \( \{ p_k \in \CP ~|~ \commits[k] \neq \perp \} \in \CQ \) \textbf{then} \label{} \\
 \>\> \( V =\lbrace \commits[k] ~|~ \commits[k] \neq \perp \rbrace \); \( \commitQC \gets \op{QC} \left( V \right)  \) \label{} \\
 \>\> send message \( \msg{\str{decide}}{\curView, \perp, \commitQC} \) to all $p_{j} \in \CP$  \label{}\\
 \\
 \textbf{upon} receiving a message \( \msg{decide}{\viewNumber, \node, \justify} \) from $p_\ell$ \label{}\\
 \>\>\>\textbf{such that} $ \viewNumber = \curView $ \textbf{and} $\justify.type = \str{commit}$ \textbf{do} \label{}\\
 \> \textbf{output }\( \op{decide}\left( \justify.\node \right)  \) \label{} \\
 \> send message \( \msg{\str{new-view}}{\curView, \perp, \prepareQC} \) to $p_{\ell+1}$  \label{line:new-view-interrupt}
   \end{numbertabbing}
}
\caption{Basic HotStuff, code for process $p_i$}
\label{alg:basic-hotstuff}
\end{algo*}

\begin{theorem}\label{thm:basic-hotstuff-safety}
If $w$ and $b$ are conflicting nodes, they cannot be both decided, each by a correct replica. 
\end{theorem}

\begin{proof}
Let $\qc_1$ and $\qc_2$ be the valid certificates, with $\qc_1$ created with the votes of a quorum $Q_1$ and $\qc_2$ with the votes of a quorum $Q_2$, that convinced the two replicas to decide, that is $\qc_1.\type = \str{commit}$, $\qc_1.\node = w$, $\qc_2.\type = \str{commit}$, $\qc_2.\node =b $. 
Also, let $\qc_1.\viewNumber = v_1$ and $\qc_2.\viewNumber = v_2$. First note that $v_1$ and $v_2$ cannot be the same. That would mean that the votes in $Q_1$ and $Q_2$ were cast in the same view, which would require the replicas in $Q_1 \cap Q_2$ to vote twice in that view. But this is impossible, since algorithm \ref{alg:basic-hotstuff} allows replicas to vote only once in the \emph{commit} phase $Q_1 \cap Q_2$ contains at least on correct replica. W.l.o.g. let $v_1 < v_2$ and let $v_2$ be the first view after $v_1$ for which a conflicting block is decided.

For $\qc_2$ to be created, according to algorithm \ref{alg:basic-hotstuff} there must first have been a valid prepareQC for node $b$. This could have been formed in view $v_2$ or in an earlier. Let $v_s$ be the first view after $v_1$ in which a valid prepareQC $\qc_s$ was formed. So, $\qc_s.\type = \str{prepare}$, $\qc_s.\node = b$ and $\qc_s.\viewNumber = v_s$ and $Q_s$ is a quorum of replicas, whose votes where used to create $qc_s$.

Consider now a replica $r$ that voted for $\qc_1$ and $\qc_s$, i.e. $r \in Q_1 \cap Q_s$. During view $v_1$, $r$ must had received a valid \precommitQC and set it to its \lockedQC, with $\lockedQC.\node = w$, before casting its vote for the \commitQC $\qc_1$. Let us examine now the \emph{prepare} phase of view $v_s$, in which the leader proposed the new block $b$, and specifically the conditions in lines \ref{line:safety} and \ref{line:liveness}. By the minimality of $v_s$, $r$ was still locked on $\lockedQC$ in that phase. By assumption $b$  and $w$ were conflicting nodes, so the condition in line \ref{line:safety} was \false.  Moreover, $\justify.\viewNumber$ was not larger than $\lockedQC.\viewNumber = u_1$, again by the minimality of $v_s$, because that would mean that a valid \prepareQC was created in a view smaller than $v_s$. So the condition in line \ref{line:liveness} was also \false. As a result, every replica in $r \in Q_1 \cap Q_s$ must be faulty. But this contradicts the quorum intersection property, thus such $\qc_1$ and $\qc_2$ cannot exist.
\end{proof}

\begin{theorem}\label{thm:basic-hotstuff-liveness}
After GST, there exists a bounded time period $T_f$ such that if all correct replicas remain in view $v$ during $T_f$ and the leader for view $v$ is correct, then a decision is reached.
\end{theorem}

\begin{proof}
Assume a correct leader that collects \str{new-view} messages from a quorums $Q_1$ of replicas. Let $\qc_l$ be the highest \lockedQC among all replicas. There must be at least a quorum $Q_2$ of replicas that have received (and voted for) a prepareQC $\qc_p$ that matches $\qc_l$. By the quorum intersection property, $Q_1 \cap Q_2$ contains a non-empty set of non-faulty replicas, through which the leader will learn $\qc_p$ and use it as its \highQC in the \str{prepare} message. Since all the correct replicas remain in view $v$, they will vote in all the phases and a decision will be reached.
\end{proof}

In Appendix~\ref{sec:impl_hs}, Algorithm~\ref{alg:implemented-hotstuff} we also show the generalized \emph{implemented} HotStuff, so as to document our changes with regard to~\cite{DBLP:conf/podc/YinMRGA19}.

\section{Evaluation}
\label{sec:eval}

We have implemented general BQS in
HotStuff~\cite{DBLP:conf/podc/YinMRGA19}\footnote{We used the
  prototype implementation available at
  https://github.com/hot-stuff/libhotstuff.}.
The new functionality has been added in the form of a C++ library into the existing code base. We use \emph{nholmann-json}~\cite{JSON} to parse the user-defined quorum-specification file
and Shoup's \emph{NTL}~\cite{NTL} for linear algebra over $\BZ_p$.
As in the original version of HotStuff, our implementation uses secp256k1 for all signatures. The prototype code does not make use of threshold signatures, instead stores all the received votes for a block and verifies them independently. We keep the same logic for our generalized quorum votes.

\dotparagraph{Setup}
In our evaluations, we report on benchmarks with four different versions of HotStuff that differ in the way how replicas and clients encode quorums. Their features are summarized in Table~\ref{tab:systems}.  
In the original HotStuff algorithm (\emph{\off}), replicas and clients know the parameters $n$ and $f$, the number of total replicas and failures, respectively, and determine whether they have received messages from a quorum by counting. 
In \emph{\form} the replicas and the clients are given the Byzantine quorum system, which can be a threshold or a generalized BQS, encoded as a monotone Boolean formula. Here we use Algorithm~\ref{alg:isAuthorized_formula} to check for quorums.
For \emph{\on}, replicas and clients are given an MSP-encoded BQS, again threshold or generalized, and use the algorithm of Section \ref{sec:encode_msp} to decide whether a set of parties is a quorum.
According to the standard practice, replicas use batching to amortize various expensive operations (signatures and potentially Gaussian elimination) over multiple requests.  However, the clients collect responses individually for every single request. This incurs a large cost that is not part of the replication protocol per se but is due to the way how clients produce requests and check for quorums.
For this reason, we experiment also with a fourth protocol, called \emph{\mix}, where only the replicas use an MSP. In this setting, the clients are mapped to replicas. Since the replicas receive and verify batches of requests at once, there is no further need to perform the quorum check on individual requests.

\begin{table}
    \caption{The evaluated protocols.}
    \label{tab:systems}
 	 \centering
    \begin{tabular}{|c|c|c|c|} 
      \hline
      &\multicolumn{2}{|c|}{BQS implementation in}& Supported\\
      System & replicas 		&clients 					& types of BQS\\
      \hline
      \off 		& counting 		& counting				& threshold\\
      \form 		& MBF	 		& MBF					& threshold \& generalized\\
      \on 		& MSP 			& MSP					& threshold \& generalized\\
      \mix 		& MSP 			& counting				& threshold \& generalized\\
      \hline
    \end{tabular}
\end{table}

The evaluation in the original HotStuff paper~\cite{DBLP:conf/podc/YinMRGA19} uses a batch size of 400 because the latency of batching becomes higher than the cost of cryptographic operations with larger batches.  Hence, we run all our experiments with batch size~400.
Finally, we work only with the three-phase HotStuff.

We use VMs on a leading cloud provider, with each replica or client running on a single VM with 16 vCPUs (Intel Xeon Broadwell, 2.6~GHz, or Intel Xeon Skylake, 2.7~GHz), 32 GB RAM, and SSD local storage. We use a varying number of VMs -- up to 40 replicas and 32 clients. All experiments are done over the LAN inside one data center, with a RTT of less than 1~ms. As this setup eliminates most network delays, it exposes the overhead added by the generalized BQS code. For the same reason, we use only zero-sized request and response payloads. In realistic deployments (on a wide-area network and with significant payload data), the extra cost of generalized quorums would be less visible. All measurements are made on the client. Finally, the maximum available bandwidth among the VMs was measured by \emph{iperf} as 1--2 Gbits per second.

\dotparagraph{Throughput vs.\ latency}
We first measure throughput and latency in a small system with four replicas, with the goal of comparing the behavior of the four different quorum-system implementations.
We use a threshold BQS because all four protocols can be instantiated with it, that is, in \off, this is specified by two numbers, $n=4$ and $f=1$, in \form~by the $\Theta_3^4 (\CP)$ MBF, and in the last two protocols by an MSP implementing the $\Theta_3^4 (\CP)$ access structure.
The reported values were produced by first fixing the request rate per client and increasing the number of clients from one to eight and then, with the number of clients fixed at eight, increasing the request rate even further for each of them, until the system saturates.
The result is depicted in Figure~\ref{fig:thr_lat_4}.

All four protocols exhibit similar behavior.
\off~saturates at 188.4K tx/sec, followed by \form~at 179.3K tx/sec, which is less than 5\% lower. 
The peak throughput of MSP-based protocols are slightly lower.
Specifically, \mix~delivers 175.5K tx/sec before saturation, which translates to an overhead of almost 7\% compared to \off,
while \on~reaches roughly 167.8K tx/sec, for an overhead of 11\%.
The latency at the saturation point is about 11.5ms for all protocols.
We conclude that in a small system like this, with four parties, generalizing a protocol does not significantly impact its efficiency.

\begin{figure}\vspace*{-3mm}
\begin{center}
	\includegraphics[width=0.5\textwidth]{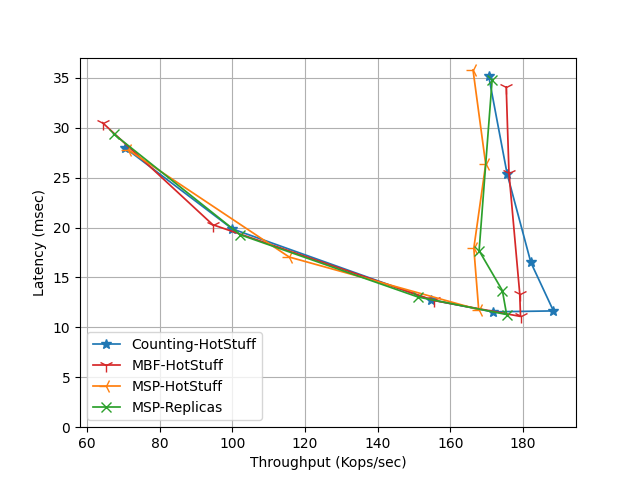}
\end{center}
\vspace*{-3mm}\caption{Throughput vs.\ latency for 1--8 clients and different implementations of the threshold BQS with four replicas.}\label{fig:thr_lat_4}
\end{figure}

\dotparagraph{Scalability}
In this evaluation we measure the throughput and latency in a system with a varying number of replicas. We use $n = 3f + 1$ replicas, for $f \in \{1, \ldots, 10\}$, and a varying number of clients. 
The trust assumption is again a threshold quorum system with $n$ replicas, of which up to $f$ may fail, specified in the appropriate way for each system. 
For each $n$ we increase the request rate per client and report the throughput and latency just before saturation. The question we want to answer is how the generalized protocols (\form, \on, \mix) scale in comparison to \off.  
The results are shown in Figures~\ref{fig:thr} (throughput) and~\ref{fig:lat} (latency).

We notice that \off~and \form~scale up almost identically. 
In a system with 31 replicas they achieve a throughput of 80.5K and 78.7K tx/sec, respectively, with latencies of 29.6ms and 26ms.
\mix~achieves throughput and latency very similar to \off~for low values of $n$ and comparable to \off~for higher $n$.
At $n=13$ the throughput of \mix~is 9\% lower than that of \off, while the latency is only 4\% higher.
With $n=31$, throughput and latency of \mix~lie both approximately 35\% behind the numbers for \off.
We conclude that the overhead added by the MSP-based quorum-checking code is relatively small for the replicas, considering all the other tasks they have to carry out, such as signature evaluation and message processing, especially when batching is used.
However, the protocol where both the replicas and the clients use MSPs does not scale so well. This is because clients do not use batching but operate on the MSP matrix for every received response. Moreover, in the original HotStuff prototype implementation, the clients do not verify the signatures on the response messages at all (!) and therefore, this operation is very fast and lets the overhead of the MSP appear large.  With signature verification enabled, as in a production system, additional cost incurred by the MSP representation would be much less visible.

\begin{figure}
	\centering
	\subfloat[Throughput]{\includegraphics[width=0.5\textwidth]{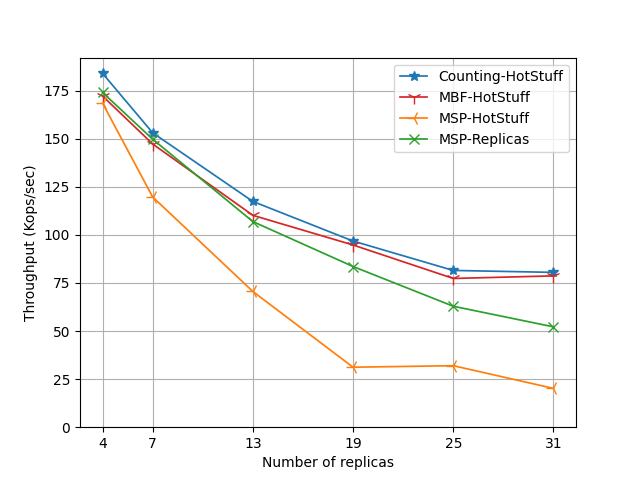}\label{fig:thr}}
        \subfloat[Latency]{\includegraphics[width=0.5\textwidth]{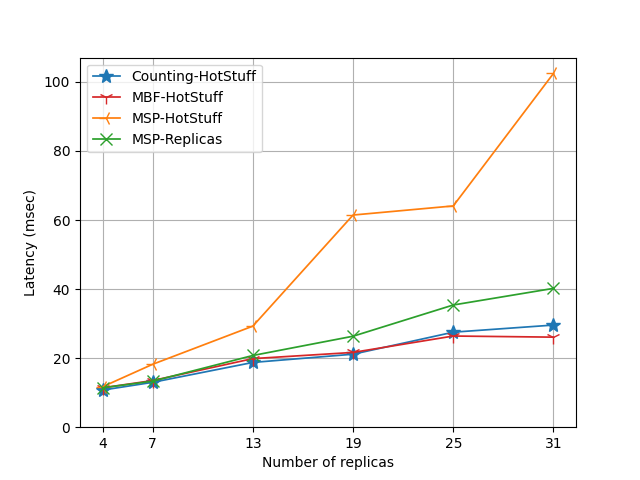}\label{fig:lat}}
	\caption{Scalability of the four protocols when instantiated with a threshold BQS.}
\end{figure}     

\dotparagraph{Scalability with generalized Byzantine quorum systems}
We now evaluate the protocols beyond threshold BQS. The question we want to answer with this benchmark is how they scale when instantiated with a generalized BQS, in comparison to when instantiated with a threshold BQS. 
We focus on \form~and \mix, which perform best in the previous experiments, and run them on two different families of BQS. The first is the 2-layered-1-common  generalized BQS presented in Section~\ref{sec:techniques}, and the second is a threshold BQS. For \onecom~we vary the parameter $k$ from 4 to 10, resulting in a system with $4k$ parties, while the threshold BQS is specified by the MBF $ \Theta_{\ceil{\frac{2n + 1}{3}}}^n \left( p_1, \ldots, p_n \right)$, for $n=4k$. 
We do not consider \off~in this benchmark because it cannot be instantiated with the generalized BQS.

We first report a direct comparison between the MBF method and the MSP method for encoding a generalized BQS. In Figure~\ref{fig:microbench} we show the memory required by each replica to store the BQS specification and the average time needed to check whether a set (chosen uniformly at random and repeated 10000 times) is a quorum, based on our implementation. Both \form~and \mix~are considered, instantiated with both the \onecom~and the threshold BQS. The MBF-based encoding is far more efficient than the MSP implementation, both in terms of memory consumption and evaluation time.

In Figures~\ref{fig:1com_thr} and~\ref{fig:1com_lat} we report the throughput and latency, respectively. In this experiment we run two replica instances in every VM, so the values reported here are overall lower than in the previous benchmarks. The performance of \mix~when running with the generalized and the threshold quorum specifications is similar. This is because in both cases the replicas have to perform Gaussian elimination on matrices of comparable dimensions. 
\form~also scales in a similar way for both families of trust assumptions, but this benchmark shows that its efficiency is slightly affected by the specified BQS.
This is, first, because generalized BQS are implemented by longer monotone Boolean formulas,
but also because generalized BQS have a (sometimes much) smaller number of quorums than threshold BQS, which might affect the leader when waiting for a quorum of votes.
It is worth to mention that the MBF-based protocols perform better than the MSP-based ones also in this benchmark.

\begin{figure}\vspace*{-3mm}
\begin{center} 	
	\includegraphics[width=0.5\textwidth]{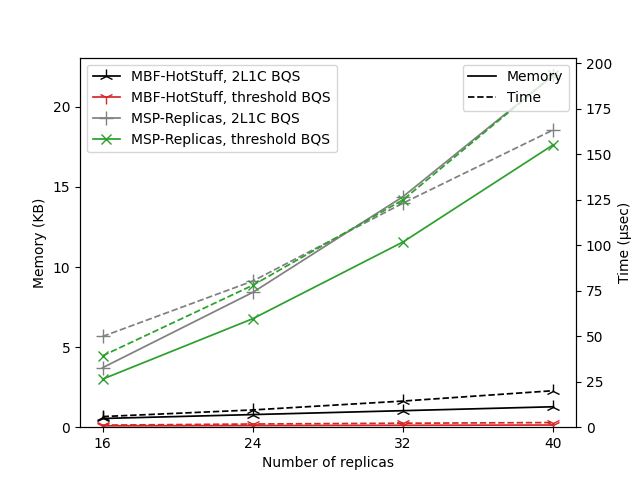}
\end{center}
\vspace*{-3mm}\caption{Memory and time required to store the BQS and check for quorums for the MSP-based and MBF-based implementations, when instantiated with two different trust assumptions, the generalized \onecom~for $k = 4, \ldots, 10$, resulting in $4k$ parties, and the $2/3$ Byzantine threshold on a set of $4k$ parties.}
\label{fig:microbench}
\end{figure}

\begin{figure}[ht]
	\centering
	\subfloat[Throughput]{\includegraphics[width=0.5\textwidth]{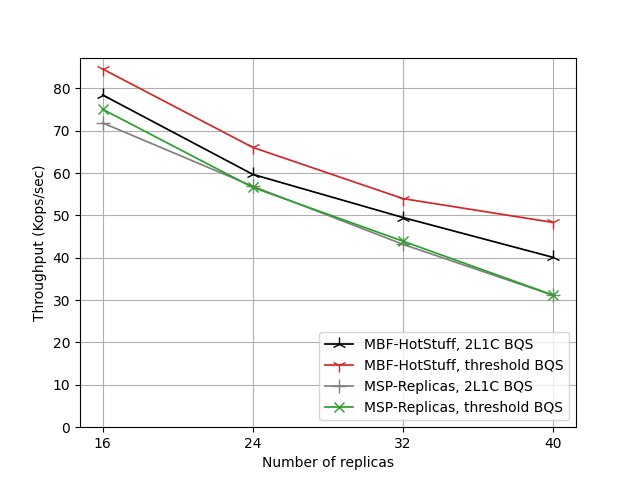}\label{fig:1com_thr}}
        \subfloat[Latency]{\includegraphics[width=0.5\textwidth]{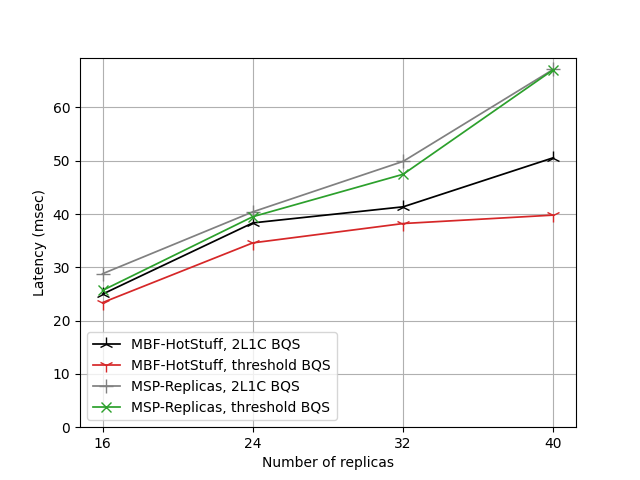}\label{fig:1com_lat}}
	\caption{Scalability of \form~and \mix~when running under two different trust assumptions, the generalized \onecom~for $k = 4, \ldots, 10$, resulting in $4k$ parties, and the $2/3$ Byzantine threshold on a set of $4k$ parties.}
\end{figure}

\dotparagraph{Discussion}
Our benchmarks illustrate the added value of generalized BQS and demonstrate that they have small overhead. One can therefore specify complex, non-threshold trust assumptions in SMR protocols without significantly sacrificing efficiency. 
The MBF-based protocol performs consistently better than the MSP-based, which can be expected due to the higher implementation complexity.
The performance of the MBF-based protocol was identical or comparable to the original threshold HotStuff, although it can be slightly affected by the complexity of the BQS, since more complex trust assumptions result in longer formulas.
The protocol where both the replicas and the clients use the MSP does not scale well and can only be used in small systems.
Nonetheless, in applications where all the nodes participate in the protocol, i.e., clients are not disjoint from servers, encoding the BQS as an MSP also results in high efficiency, as was shown by \mix~in the benchmarks. %In small systems of up to 13 nodes, particularly, the throughput of a system with generalized trust assumptions was less than 10\% lower than the throughput of a system with threshold BQS, while the latency of the two systems was almost equal. From a user perspective, the overhead induced by the generalized Byzantine quorum system is minimal, as they only have to provide one config file prior to the start of the system.
We anticipate that our work will pave the way for more protocols generalizing threshold trust assumptions. This can be combined with the novel ideas presented in the BFT literature, e.g., the combination of crash and Byzantine faults~\cite{DBLP:conf/opodis/CachinSV16} or with peer-to-peer gossip~\cite{DBLP:journals/corr/abs-1807-04938}.

\section*{Acknowledgments}

This work has been funded by the Swiss National Science Foundation (SNSF)
under grant agreement Nr\@.~200021\_188443 (Advanced Consensus Protocols).

\bibliography{references, dblpbibtex}
\bibliographystyle{ieeesort}

\vfill
\appendix

\section{Implemented HotStuff}\label{sec:impl_hs}

The implemented version of HotStuff using generalized quorums is presented in Algorithm \ref{alg:implemented-hotstuff}.
The call $\isQuorum(\votes[b])$ checks whether the replicas in $\votes[b]$ constitute a quorum, using our algorithm described in \ref{sec:encode_msp}.
The function \getLeader is not defined in HotStuff but is specified by the application. Procedure \onBeat is also called by the leader in order to propose new clients' commands at points specified by the application.

\begin{algo*}
\vbox{
\small
\begin{numbertabbing}\reset
 xxxx\=xxxx\=xxxx\=xxxx\=xxxx\=xxxx\=MMMMMMMMMMMMMMMMMMM\=\kill
 \textbf{State} \\
 \> \(\prepareQC \gets \perp\); \(\lockedQC \gets \perp\); \(\curView \gets 1 \) \\
 \\
 \textbf{procedure} \createLeaf \( \left(  \parent, \cmd, \qc, \height  \right) \) \label{}\\
 \> \( b \gets \text{ new node}  \) \label{}\\ 
 \> \( b.\parent \gets \parent \); \( b.\cmd \gets \cmd \) \label{}\\
 \> \( b.\justify \gets \qc \); \( b.\height \gets \height \) \label{} \\
 \>\textbf{return} $b$\label{} \\
 \\
 \textbf{procedure} \update \( \left(  b^\ast  \right) \)\label{} \\
 \> \( b^{\prime\prime} \gets b^\ast.\justify.\node  \); \( b^{\prime} \gets b^{\prime\prime}.\justify.\node  \); \( b \gets b^{\prime}.\justify.\node  \) \label{}\\ 
 \> \( \updateQCHigh \left(  b^\ast.\justify  \right) \) //PRE-COMMIT phase on $ b^{\prime\prime}$ \label{}\\
 \>\textbf{if} \( b^\prime.\height > b_{\text{lock}}.\height \) \textbf{then} //COMMIT phase on $ b^{\prime}$ \label{}\\
 \>\> \( b_{\text{lock}} \gets b^\prime \) \label{}\\
 \>\textbf{if} \( b^{\prime\prime}.\parent = b^{\prime} \textbf{ and } b^{\prime}.\parent = b \) \textbf{then} //DECIDE phase on $ b$ \label{}\\
 \>\> \( \onCommit(b) \) \label{}\\
 \>\> \( b_{\text{exec}} \gets b \) \label{}\\
 \\
 \textbf{procedure} \onCommit \( \left(  b  \right) \) \label{}\\
 \>\textbf{if} \( b_{\text{exec}}.\height <  b.\height \) \textbf{then} \label{}\\
 \>\> \( \onCommit (b.\parent) \)\label{} \\
 \> \( \execute(b.\cmd) \) \label{}\\
 \\
 \textbf{procedure} \onReceiveProposal \( \left(  m = \msg{\str{generic}}{b_{\text{new}}, \perp} \right) \) \label{}\\
 \>\textbf{if} \( b_{\text{new}}.\height > \var{vheight} \textbf{ and } 
 \left( b_{\text{new}} \text{ extends } b_{\text{lock}} \textbf{ or } b_{\text{new}}.\justify.\node.\height > b_{\text{lock}}.\height \right) \) \textbf{then} \label{}\\
 \>\> \( \var{vheight} \gets b_{\text{new}}.\height\)\label{} \\
 \>\> send message \(  \msg{\str{generic-vote}}{b_{\text{new}}, \perp}  \) to $\getLeader()$ \label{} \\
 \> \( \update(b_{\text{new}}) \) \label{}\\
 \\
 \textbf{procedure} \onReceiveVote \( \left( m = \msg{\str{generic-vote}}{b, \perp} \right) \) \text{ from } $p_j$ \label{}\\
 \> \( \votes[b] \gets \votes[b] \cup \{ \langle j, m.\sig \rangle \} \)\label{} \\
 \>\textbf{if } \( \isQuorum(\votes[b]) \) \textbf{then} \label{}\\
 \>\> \( \qc \gets \QC \left( \lbrace v_j \rbrace_{j=1}^k \right) \)\label{} \\
 \>\> \( \updateQCHigh(\qc) \) \label{}\\
 \\
 \textbf{function} \onPropose \( \left(   b_{\text{new}}, \cmd, \qc_{\text{high}}  \right) \) \label{}\\
 \> \(  b_{\text{new}} \gets \createLeaf  \left(   b_{\text{leaf}}, \cmd, \qc_{\text{high}}, b_{\text{leaf}}.\height + 1  \right)  \) \label{}\\
 \> send message \( \msg{\str{generic}}{ b_{\text{new}}, \perp} \) to all $p_{j} \in \CP$ \label{} \\
 \>\textbf{return} \( b_{\text{new}} \) \label{}\\
 \\
 \textbf{procedure} \( \updateQCHigh \left( \qc_{\text{high}}^\prime \right) \) \label{}\\
 \>\textbf{if } \(  \qc_{\text{high}}^\prime .\node.\height >  \qc_{\text{high}}.\node.\height \) \textbf{then} \label{} \\
 \>\> \(  \qc_{\text{high}} \gets  \qc_{\text{high}}^\prime  \) \label{}\\
 \>\> \( b_\text{leaf} \gets  \qc_{\text{high}}.\node \) \label{}\\
 \\ 
 \textbf{procedure} \( \onBeat \left( \cmd \right) \) \label{}\\
 \>\textbf{if } \(  i = \getLeader() \) \textbf{then} \label{}\\
 \>\> \( b_\text{leaf} \gets \onPropose \left(  b_\text{leaf}, \cmd, \qc_{\text{high}} \right) \) \label{}\\
 \\ 
 \textbf{procedure} \( \onNextSyncView \left( \cmd \right) \)\label{} \\
 \> send message \( \msg{\str{new-view}}{\perp, \qc_{\text{high}} } \) to $\getLeader()$  \label{}\\
 \\
 \textbf{procedure} \( \onReceiveNewView \left( \msg{\str{new-view}}{\perp,  \qc_{\text{high}}^\prime } \right) \)\label{} \\
 \> \( \updateQCHigh \left( \qc_{\text{high}}^\prime \right) \) \label{}
 \end{numbertabbing}
}
\caption{Implemented HotStuff, code for process $p_i$}
\label{alg:implemented-hotstuff}
\end{algo*}

\end{document}